\documentclass[12pt, english, a4paper]{amsart}
\usepackage{ragged2e}

\usepackage{graphicx}

\usepackage[applemac]{inputenc} 
\usepackage{polynom}

\usepackage{amsmath,amssymb}
\usepackage{bm}
\usepackage{euscript,color}
\usepackage{babel}
\usepackage{amsthm}
\usepackage{amsfonts}
\usepackage{latexsym}
\usepackage{fancyhdr}
\usepackage{enumerate}

\usepackage[numbers]{natbib}

\usepackage[left=2.5cm, top=3cm, bottom=3cm, right=2.5cm]{geometry}
\newcommand{\Q}{\mathbb{Q}}
\newcommand{\R}{\mathbb{R}}
\newcommand{\N}{\mathbb{N}}

\newcommand{\comment}[1]{}

\makeatletter
\renewcommand{\section}{\@startsection%
{section}% name
{1}% level
{0mm}% indent
{1.5\bigskipamount}% beforeskip
{0.5\bigskipamount}% afterskip
{\centering\normalsize\sc}}% style

\renewcommand{\paragraph}{\@startsection%
{paragraph}% name
{4}% level
{0mm}% indent
{\bigskipamount}% beforeskip
{-1.25ex}% afterskip
{\normalsize\sl}}% style

\def\provedboxcontents#1{$\square$}

\makeatother

\newtheoremstyle{thm}{}{}{\slshape}{}{\scshape}{.}{0.5em}{}

\newtheoremstyle{def}{}{}{}{}{\scshape}{.}{0.5em}{}

\newtheoremstyle{rmk}{}{}{}{}{\scshape}{.}{0.5em}{}

\newtheoremstyle{claim}{}{}{}{}{\slshape}{.}{0.5em}{}

\theoremstyle{thm}
\newtheorem{newstatement}{newstatement}
\newtheorem{lemma}[newstatement]{Lemma}
\newtheorem{theorem}[newstatement]{Theorem}

\newtheorem*{conjecture*}{Conjecture}

\theoremstyle{def}

\theoremstyle{rmk}
\newtheorem{remark}[newstatement]{Remark}

\theoremstyle{claim}

\expandafter\let\expandafter\oldproof\csname\string\proof\endcsname
\let\oldendproof\endproof
\renewenvironment{proof}[1][\proofname]{%
  \oldproof[\slshape #1]%
}{\oldendproof}

\let\geq\geqslant
\let\leq\leqslant

\let\phi\varphi
\let\epsilon\varepsilon

\renewcommand{\emph}[1]{{\slshape #1}}

\title{Endowments, Patience Types, and Uniqueness in Two-Good HARA Utility Economies}

\author{Andrea Loi}
\address{Andrea Loi, Dipartimento di Matematica e Informatica \\
         Universit\`a di Cagliari, Italy.}
         \email{loi@unica.it}

\author{Stefano Matta}
\address{Stefano Matta, Dipartimento di Scienze economiche e Aziendali \\
         Universit\`a di Cagliari, Italy.}
         \email{smatta@unica.it}

\date{}

\thanks{The first author was  supported  by INdAM. GNSAGA - Gruppo Nazionale per le Strutture Algebriche, Geometriche e le loro Applicazioni and
by KASBA, funded by Regione Autonoma della Sardegna.
Both authors were supported 
by STAGE, funded by Fondazione di Sardegna.}

\begin{document}
\begin{abstract}
This paper establishes a link between endowments, patience types, and the parameters of the HARA Bernoulli utility function that ensure equilibrium uniqueness in an economy with two goods and two impatience types with additive separable preferences. 
We provide sufficient conditions that guarantee uniqueness of equilibrium for any possible value of $\gamma$
in the HARA  utility function $\frac{\gamma}{1-\gamma}\left(b+\frac{a}{\gamma}x\right)^{1-\gamma}$.
The analysis contributes to the literature on uniqueness in pure exchange economies with two-goods and
two agent types and extends the result in \cite{lm22}.
\end{abstract}
\maketitle

\vspace{0.3in}


\vspace{0.3in}


\vspace{0.3in}

\section{Introduction}

The relationship between  risk aversion, the number of consumer types $I$ and the uniqueness of the equilibrium price has been analysed in a recent article \cite{lm22}.
 More precisely, in
an economy with two goods and an arbitrary number $I$ of impatience types, where each type has additive separable preferences with a HARA Bernoulli utility function $u_H(x):=\frac{\gamma}{1-\gamma}\left(b+\frac{a}{\gamma}x\right)^{1-\gamma}$, it has been shown that the equilibrium is unique
if the parameter $\gamma$ is in the range $\left(1,\frac{I}{I-1}\right]$.

While it is well known \citep{helo,msc91,msc95} the effect on uniqueness when $\gamma$ takes a value between $0$ and $1$ (in the case $\gamma=1$ the function becomes logarithmic), it is not known what conditions ensure uniqueness when $\gamma$ is greater than $2$.
In this perspective, \cite{lm22} analysed only the particular case when $\gamma=3$ and $I=2$, and found sufficient conditions that ensure uniqueness. 

It is a natural question to ask whether a similar result can be found outside the above interval. 
More specifically, this would mean finding sufficient conditions that guarantee uniqueness of the equilibrium for any value of the $\gamma$ parameter. This is related to \cite{gewa}'s remarks on the difficulty of finding a sufficient condition, expressible in closed form, that would allow, for DARA-type utilities, the introduction of more heterogeneity among agents in order to overcome the restrictive assumption of identical endowments to ensure uniqueness (see \cite[Proposition 2]{gewa}).

This paper provides a positive answer to the question above for HARA utilities, an important subclass of the DARA type. More precisely,
our main result, Theorem \ref{mainteor}, shows the connection between endowments, patience types and the parameters of the HARA utility function  that ensure the uniqueness of the equilibrium.
To obtain this result, we will follow the approach of \cite{lm22}, where
the excess demand function is approximated by a polynomial whose variable, the price, is raised to a power
dependent on $\gamma$.  An algebraic result, Lemma \ref{divpol}, which links the existence of a double root of a polynomial to an inequality involving its coefficients, allows us to use a topological argument
to prove our main result.

For an overview of the literature on uniqueness, in addition to the well-known contributions by 
 \cite{ke98} and \cite{msc91}, we also refer the reader to the two recent contributions by \cite{gewa} and \cite{towa}
 for two-good, two-agent pure exchange economies.

This short note is organised as follows. Section \ref{prel} analyses the economic setting using the polynomial approach. Section \ref{mainsec} proves our main result.

\section{Preliminaries}\label{prel}

Consider an economy with two goods and $I=2$ impatience types,
where  type  $i$ has preferences represented by the utility function
\begin{equation}\label{addsep}
u_i(x,y)=u_H(x)+\beta_iu_H(y),
\end{equation}
where   $u_H$ is HARA, i.e.
\begin{equation}\label{uH}
u_H(x):=\frac{\gamma}{1-\gamma}\left(b+\frac{a}{\gamma}x\right)^{1-\gamma},\ \gamma>0, \gamma \neq 1, a>0,  b\geq 0.
\end{equation}

Let $\epsilon$ be a rational number $\frac{m}{n}$, $m,n\in\N$ sufficiently close to $\frac{1}{\gamma}$.
Suppose that $\gamma>2$ and, hence, $n>2m$.
Denoting by $(e_i,f_i)$ consumer i's endowments,  the standard maximisation problem over the budget constraint $px_i+y_i\leq pe_i+f_i$ gives (see \cite[formula (14)]{lm22})
the aggregate excess demand function for good $x$:

\begin{equation}\label{z1}
\sum_{i=1}^2\frac{b-bp^\epsilon\sigma_i+a\epsilon\left(pe_i+f_i\right)}{a\epsilon\left(p+\sigma_i p^\epsilon \right)}-(e_1+e_2),
\end{equation}
where
$$\epsilon\approx \frac{1}{\gamma}, \ \sigma_i:=\beta_i^{\epsilon}, \ i=1,2.$$

Following \cite{lm22}, we combine terms over a common denominator and take the numerator,
then we collect terms in $p$, divide by $p^{\epsilon}$, and we get:
\begin{equation*}
\begin{split}
&p (-a e_1 \sigma_1 \epsilon -a e_2 \sigma_2 \epsilon -b \sigma_1-b \sigma_2)+p^{1-\epsilon } (a f_1 \epsilon +a f_2 \epsilon +2 b)+
\\
& p^{\epsilon } (-a e_1 \sigma_1 \sigma_2 \epsilon -a e_2 \sigma_1 \sigma_2 \epsilon -2 b \sigma_1 \sigma_2)+a f_1 \sigma_2 \epsilon +a f_2 \sigma_1 \epsilon +b \sigma_1+b \sigma_2
\end{split}
\end{equation*}

\noindent Recalling that $\epsilon=\frac{m}{n}$ and by letting, with a slight abuse of notation, $x:=p^{1/n}$, we rewrite the previous expression in decreasing order as follows:
\begin{equation}\label{pol_can}
A(e, \sigma, a, b)x^n + B(f, \sigma, a, b)x^{n-m} + C(e, \sigma, a, b)x^m + D(f, \sigma, a, b),
\end{equation}
where

\begin{equation}\label{ABCD}
\begin{alignedat}{3}
A(e, \sigma, a, b):=& -(e_1\sigma_1+e_2\sigma_2)-\frac{b}{a\epsilon}(\sigma_1+\sigma_2)<0,\\
B(f, \sigma, a, b):=& (f_1+f_2)+\frac{2b}{a\epsilon}>0,\\ 
C(e, \sigma, a, b):=& (e_1+e_2)\sigma_1\sigma_2-\frac{2b}{a\epsilon}\sigma_1\sigma_2<0, \\ 
D(f, \sigma, a, b):=& (f_1\sigma_2+f_2\sigma_1)+\frac{b}{a\epsilon}(\sigma_1+\sigma_2)>0.
\end{alignedat}
\end{equation}

\begin{lemma}\label{lemma_abcd}
If the following conditions hold
\begin{equation}\label{c1}
\beta_1<\beta_2,\, e_1\leq e_2,\, f_1\geq f_2,
\end{equation}
\begin{equation}\label{c2}
b\geq \frac{a}{\gamma}\left(\frac{\beta_2}{\beta_1}\right)^\frac{2}{\gamma}(e_2+f_1),
\end{equation}
then the polynomial \eqref{pol_can} satisfies the inequality
\begin{equation}\label{ineq}
A(e, \sigma, a, b)D(f, \sigma, a, b)-B(f, \sigma, a, b)C(e, \sigma, a, b)<0
\end{equation}
\end{lemma}

\begin{proof}

We will follow, mutatis mutandis, the same line of reasoning of the proof of \cite[Theorem 2]{lm22}, the only difference here is that  we deal with an arbitrary value of $\gamma$. 
The formula
$A(e, \sigma, a, b)D(f, \sigma, a, b)-B(f, \sigma, a, b)C(e, \sigma, a, b)$ can be written as

$$(\sigma_2-\sigma_1)(e_1f_2\sigma_1-e_2f_1\sigma_2) +E(e,f, \sigma, a, b),$$
where 
$$E(e,f, \sigma, a, b):=-\frac{b^2}{a^2\epsilon^2}\left(\sigma_1-\sigma_2\right)^2+\frac{b}{a\epsilon}\left[\left(e_1+e_2+f_1+f_2\right)\sigma_1\sigma_2-(e_1+f_2)\sigma_1^2 -(e_2+f_1)\sigma_2^2\right].$$

Observe that, by condition \eqref{c1}, 
$$(\sigma_2-\sigma_1)(e_1f_2\sigma_1-e_2f_1\sigma_2)\leq (\sigma_2-\sigma_1)f_1(e_1\sigma_1-e_2\sigma_2)< (\sigma_2-\sigma_1)f_1\sigma_2(e_1-e_2)\leq 0.$$

Moreover, $E(e,f, \sigma, a, b)\leq 0$ if and only if 
$$b\geq a\epsilon \frac{\left[(e_1+e_2+f_1+f_2)\sigma_1\sigma_2-(e_1+f_2)\sigma_1^2 -(e_2+f_1)\sigma_2^2\right]}{\left(\sigma_1-\sigma_2\right)^2}.$$

Again, by \eqref{c1}, we can write

\begin{align*}
a\epsilon\frac{\left[(e_1+e_2+f_1+f_2)\sigma_1\sigma_2-(e_1+f_2)\sigma_1^2 -(e_2+f_1)\sigma_2^2\right]}{\left(\sigma_1-\sigma_2\right)^2}& \\
&<a\epsilon\frac{\left[(e_1+e_2+f_1+f_2)\sigma_1\sigma_2\right]}{\left(\sigma_1-\sigma_2\right)^2}\\&<
a\epsilon\left(\frac{\sigma_2}{\sigma_1}\right)^{2}(e_2+f_1).
\end{align*}

\noindent Thus, since  $\sigma_i=\beta_i^{\epsilon}$, $i=1, 2$ and $\gamma=\frac{1}{\epsilon}$, the  proof of the lemma follows.
\end{proof}

\section{Main result}\label{mainsec}

\noindent In this section we present our main result, Theorem \ref{mainteor}.
As far as uniqueness is concerned,
we will assume an arbitrary $\gamma>2$. In fact, 
the case $\gamma\in (1,2]$ is a particular case of \cite[Theorem 1]{lm22}, while the case $\gamma\leq 1$ 
is a well known result in the literature \cite{helo,msc91,msc95}.

Observe that the zero set of aggregate demand function amounts to studying the zeros of polynomial \eqref{pol_can}. In fact, according to \cite{lm22}'s approach it is possible to approximate $\gamma$ with a rational number, since $\Q$ is dense in $\R$,  in such a way that the cardinality of the set of regular equilibria does not decrease \cite[Lemma 9]{lm22}. To provide a geometric insight, this corresponds to small perturbations of the aggregate demand function that do not allow a decrease in the number of the equilibria.

\begin{theorem}\label{mainteor}

In an economy with two goods and two impatient types with Hara preferences \eqref{addsep},
if the conditions \eqref{c1} and \eqref{c2} hold, then the equilibrium price is unique.
\end{theorem}

\begin{remark}\rm
For the general type of DARA,  \cite{gewa} observe there is not a closed-form expression that ensures uniqueness, but conditions \eqref{c1} and \eqref{c2} represent a closed-form expression for HARA utilities, an important subclass of utilities of type DARA. They are the same as those presented in \cite{lm22}, here suitably generalised.
\end{remark}

\begin{proof}
By \cite[Theorem 1]{lm22})  there exists uniqueness of equilibrium
if and only if the polynomial \eqref{pol_can}, $P(x)$, has a unique positive root.
We will prove that the inequality \eqref{ineq}, which holds by Lemma \ref{lemma_abcd}, implies
that $P(x)$ has a unique positive root. 
Assume by contradiction that $P(x)$ has more than one positive root. 
Since  $P(x)$ belongs to the path-connected  space of polynomials $Ax^n+Bx^{n-m}+Cx^m+D$, with non zero coefficients
such that $AD-BC<0$, it follows 
by the continuous dependence of the roots of a polynomial on its coefficients that $P(x)$ has indeed  a double positive root. 
Hence one can achieve the conclusion of Theorem \ref{mainteor} by the following algebraic lemma.
\end{proof}

\begin{lemma}\label{divpol}
If the polynomial \eqref{pol_can}, $P(x)={Ax^n+Bx^{n-m}+Cx^m+D}$, $ABCD\neq 0$, has a double positive root, then $AD-BC\geq0$.
\end{lemma}

\begin{proof}
By contradiction, let $\alpha>0$ be a double root of $P(x)$, that is,
 $(x-\alpha)^2$ divides $P(x)$. The following table shows the pattern of the remainders after the first $k$ steps of the division.
 
 \vspace{0.3cm}

\begin{center}
\begin{tabular}{|p{0.5in}|p{5.5in}|}\hline
{\bf Step}  & \Centering{\bf Remainder}\\
\hline
\Centering{$1$} & \Centering{$2\alpha Ax^{n-1}-\alpha^2Ax^{n-2}+Bx^{n-m}+Cx^m+D$} \\
\Centering{$2$} & \Centering{$3\alpha^2Ax^{n-2}-2\alpha^3Ax^{n-3}+Bx^{n-m}+Cx^m+D$}\\
\Centering{$3$} & \Centering{$4\alpha^3Ax^{n-3}-3\alpha^4Ax^{n-4}+Bx^{n-m}+Cx^m+D$}\\
\Centering{$\vdots$} & \Centering{$\vdots$}\\
\Centering{$k$} & \Centering{$(k+1)\alpha^kAx^{n-k}-k\alpha^{k+1}Ax^{n-k-1}+Bx^{n-m}+Cx^m+D$}\\
\hline
\end{tabular}
\end{center}
 \vspace{0.3cm}
\noindent From $n-k-1=n-m$, we get $m=k+1$ and then we can rewrite the reminder accordingly:

$$m\alpha^{m-1}Ax^{n-m+1}+ [B-(m-1)\alpha^mA]x^{n-m}+Cx^m+D.$$

\noindent Continuing the division with this new reminder, 
the next table reveals again the following pattern:

 \vspace{0.3cm}
\begin{center}
\begin{tabular}{|p{0.5in}|p{5.5in}|}\hline
\Centering{\bf Step}  & \Centering{\bf Remainder}\\
\hline
\Centering{$1$} & $[B+(m+1)\alpha^mA]x^{n-m}-m\alpha^{m+1}Ax^{n-m-1}+Cx^m+D$\\
\Centering{$2$} & $[2\alpha B+(m+2)\alpha^{m+1}A]x^{n-m-1}-[\alpha^2B+(m+1)\alpha^{m+2}A]x^{n-m-2}+Cx^m+D$\\
\Centering{$3$} & $[3\alpha^2B+(m+3)\alpha^{m+2}A]x^{n-m-2}-[2\alpha^3B+(m+2)\alpha^{m+3}A]x^{n-m-3}+Cx^m+D$\\
\Centering{$\vdots$} & \Centering{$\vdots$}\\
\Centering{$k$} & $k\alpha^{k-1}B+(m+k)\alpha^{m+k-1}A]x^{n-m-k+1}-[(k-1)\alpha^kB+(m+k-1)\alpha^{m+k}A)]x^{n-m-k}+Cx^m+D$\\
\hline
\end{tabular}
\end{center}
 \vspace{0.3cm}
\noindent From $n-m-k=m$, we get $k=n-2m$. We can rewrite the reminder as follows:

\begin{equation*}
\begin{split}
& [(n-2m)\alpha^{n-2m-1}B+(n-m)\alpha^{n-m-1}A]x^{m+1}-\\
&-[(n-2m-1)\alpha^{n-2m}B+(n-m-1)\alpha^{n-m}A]x^m+Cx^m+D,
\end{split}
\end{equation*}

\noindent that, reordering terms, becomes

\begin{equation*}
\begin{split}
& [(n-2m)\alpha^{n-2m-1}B+(n-m)\alpha^{n-m-1}A]x^{m+1}+\\ 
& +[C-(n-2m-1)\alpha^{n-2m}B-(n-m-1)\alpha^{n-m}A]x^m+D.
\end{split}
\end{equation*}
 
\noindent Starting with this new reminder, the last pattern is suggested by the following table:

 \vspace{0.3cm}
\begin{center}
\begin{tabular}{|p{0.5in}|p{5.5in}|}\hline
\Centering{\bf Step}  & \Centering{\bf Remainder}\\
\hline
\Centering{$1$} & $[C+(n-2m+1)\alpha^{n-2m}B+(n-m+1)\alpha^{n-m}A]x^m-[(n-2m)\alpha^{n-2m+1}B+
(n-m)\alpha^{n-m+1}A]x^{m-1}+D$\\
\Centering{$2$} & $[2\alpha C+(n-2m+2)]\alpha^{n-2m+1}B+(n-m+2)\alpha^{n-m+1}A]x^{m-1}-[\alpha^2C+(n-2m+1)\alpha^{n-2m+2}B+(n-m+1)\alpha^{n-m+2}A]x^{m-2}+D$\\
\Centering{$3$} & $[3\alpha^2 C+(n-2m+3)]\alpha^{n-2m+2}B+(n-m+3)\alpha^{n-m+2}A]x^{m-2}-[2\alpha^3C+(n-2m+2)\alpha^{n-2m+3}B+(n-m+2)\alpha^{n-m+3}A]x^{m-2}+D$\\
\Centering{$\vdots$} & \Centering{$\vdots$}\\
\Centering{$k$} & $[k\alpha^{k-1} C+(n-2m+k)]\alpha^{n-2m+k-1}B+(n-m+k)\alpha^{n-m+k-1}A]x^{m-k+1}-[(k-1)\alpha^kC+(n-2m+k-1)\alpha^{n-2m+k}B+(n-m+k-1)\alpha^{n-m+k}A]x^{m-k}+D$\\
\hline
\end{tabular}
\end{center}
 \vspace{0.3cm}
\noindent After $k=m$ divisions, the remainder reduces to a first degree polynomial:
$$m\alpha^{m-1}C+(n-m)\alpha^{n-m-1}B+n\alpha^{n-1}A]x-[(m-1)\alpha^mC+(n-m-1)\alpha^{n-m}B+(n-1)\alpha^nA]+D.$$

\noindent Under the hypothesis that $(x-\alpha)^2$ divides $P(x)$, the coefficients of this last remainder must vanish, that is:
\begin{equation*}
    \begin{cases}
     m\alpha^{m-1} C=-(n-m)\alpha^{n-m-1}B-n\alpha^{n-1}A\\
     D=(m-1)\alpha^mC+(n-m-1)\alpha^{n-m}B+(n-1)\alpha^nA.\\
     \end{cases}       
\end{equation*}

\noindent Multiplying the second equation by $m\alpha^{m-1}$, we get

$$m\alpha^{m-1}D=(m-1)\alpha^mm\alpha^{m-1}C+m(n-m-1)\alpha^{n-1}B+m(n-1)\alpha^{n+m-1}A,$$

\noindent where, substituting $m\alpha^{m-1} C$ with the RHS of the first equation and multiplying by $A$, we obtain

$$m\alpha^{m-1}AD=(n-2m)\alpha^{n-1}AB+(n-m)\alpha^{n+m-1}A^2.$$

\noindent Moreover, we observe that

$$m\alpha^{m-1}BC=-(n-m)\alpha^{n-m-1}B^2-n\alpha^{n-1}AB.$$

\noindent We can then write

$$m\alpha^{m-1}(AD-BC)=(n-m)\alpha^{n-m-1}(\alpha^{2m}A^2+B^2)+2(n-m)\alpha^{n-1}AB,$$

\noindent as

$$(n-m)\alpha^{n-m-1}(\alpha^{2m}A^2+B^2+2\alpha^mAB),$$

\noindent or, equivalently,

$$(n-m)\alpha^{n-m-1}(\alpha^mA+B)^2.$$

\noindent Hence, we have

$$AD-BC=\frac{n-m}{m}\alpha^n(\alpha^nA+B)^2\geq0,$$

yielding the desired contradiction.

\end{proof}

\begin{remark}\rm
It should be possible to give a more elegant proof
of Lemma \ref{divpol}  by using an approach based on 
the discriminant of a polynomial instead of a division algorithm
as in our proof. However, this alternative 
approach seems to lead to very complicated calculations that the authors were unable to handle.
\end{remark}

\end{document}